\documentclass[]{amsart}

\usepackage{cite}

\usepackage{amsmath}
\usepackage{amsfonts}
\usepackage{amssymb}
\usepackage{amsthm}
\usepackage{braket}

\usepackage{tikz}
\usetikzlibrary{positioning}
\usetikzlibrary{arrows}

\usepackage{rotating}

\newtheorem{theorem}{Theorem}

\newtheorem{corollary}[theorem]{Corollary}
\newtheorem{definition}[theorem]{Definition}

\newtheorem{lemma}[theorem]{Lemma}
\numberwithin{theorem}{section}

\theoremstyle{remark}
\newtheorem{remark}[theorem]{Remark}
\newtheorem{example}[theorem]{Example}

\DeclareMathOperator{\Span}{Span}

\def\ot{\otimes}
\def\V{\mathcal{V}}
\def\be{\mathbf{e}}

\def\picA{\tikz[baseline=4ex,scale=0.50]{
	\draw (0,0) rectangle (4.5,2.5);
	\foreach \x in {0,...,4}
	\foreach \y in {0,3}
	\draw (\x+0.25,\y-0.5) -- (\x+0.25,\y);
	\foreach \x in {0,5}
	\foreach \y in {0,...,2}
	\draw (\x-0.5,\y+0.25) -- (\x,\y+0.25);
	\node at (2.25, 1.25)    {$T$};}
}

\def\picB{\tikz[baseline=4ex,scale=0.50]{
		\draw (0,0) rectangle (4.5,2.5);
		\foreach \x in {0,...,3}
		\foreach \y in {0,3}
		\draw (\x+0.25,\y-0.5) -- (\x+0.25,\y);
		\foreach \x in {0,5}
		\foreach \y in {0,...,2}
		\draw (\x-0.5,\y+0.25) -- (\x,\y+0.25);
	    \node at (2.25, 1.25)    {$T_1$};}
}

\def\picC{\tikz[baseline=4ex,scale=0.50]{
		\draw (0,0) rectangle (4.5,2.5);
		\foreach \y in {0,3}
		\draw (4.25,\y-0.5) -- (4.25,\y);
	    \node at (2.25, 1.25)    {$T_2$};}
}

\def\picD{\tikz[baseline=4ex,scale=0.50]{
		\foreach \x in {0,...,3}
		\foreach \y in {0,...,2}
		\draw (\x,\y) rectangle (\x+0.5,\y+0.5);
		\foreach \x in {0,...,3}
		\foreach \y in {0,...,3}
		\draw (\x+0.25,\y-0.5) -- (\x+0.25,\y);
		\foreach \x in {0,...,5}
		\foreach \y in {0,...,2}
		\draw (\x-0.5,\y+0.25) -- (\x,\y+0.25);
		\foreach \y in {0,...,2}
		\draw (4,\y+0.25) -- (4.5,\y+0.25);}
}

\def\picE{\tikz[baseline=4ex,scale=0.50]{
		\draw (1,0) rectangle (4.5,2.5);
		\foreach \y in {0,...,2}
		\draw (0,\y) rectangle (0.5,\y+0.5);
		\foreach \x in {0,...,3}
		\foreach \y in {0,3}
		\draw (\x+0.25,\y-0.5) -- (\x+0.25,\y);
		\foreach \x in {0,1,5}
		\foreach \y in {0,...,2}
		\draw (\x-0.5,\y+0.25) -- (\x,\y+0.25);
		\foreach \y in {1,2}
		\draw (0.25,\y-0.5) -- (0.25,\y);
		\node at (2.75, 1.25)    {$B_{v_1}$};}
	}

\def\picF{\tikz[baseline=4ex,scale=0.50]{
		\draw (1,0) rectangle (4.5,2.5);
		\foreach \y in {0,...,2}
		\draw (0,\y) rectangle (0.5,\y+0.5);
		\foreach \x in {0,...,4}
		\foreach \y in {0,3}
		\draw (\x+0.25,\y-0.5) -- (\x+0.25,\y);
		\foreach \x in {0,1,5}
		\foreach \y in {0,...,2}
		\draw (\x-0.5,\y+0.25) -- (\x,\y+0.25);
		\foreach \y in {1,2}
		\draw (0.25,\y-0.5) -- (0.25,\y);
		\node at (2.75, 1.25)    {$B_{v_1}\otimes T_2$};}
}

\def\picG{\tikz[baseline=4ex,scale=0.50]{
		\foreach \x in {0,...,4}
		\foreach \y in {0,...,2}
		\draw (\x,\y) rectangle (\x+0.5,\y+0.5);
		\foreach \x in {0,...,4}
		\foreach \y in {0,...,3}
		\draw (\x+0.25,\y-0.5) -- (\x+0.25,\y);
		\foreach \x in {0,...,5}
		\foreach \y in {0,...,2}
		\draw (\x-0.5,\y+0.25) -- (\x,\y+0.25);}
	}

\def\picH{\tikz[baseline=4ex,scale=0.50]{
		\foreach \x in {0,1}
		\foreach \y in {0}
		\draw (\x,\y) rectangle (\x+0.5,\y+0.5);
		\foreach \x in {0,1,2}
		\foreach \y in {0}
		\draw (\x-0.5,\y+0.25) -- (\x,\y+0.25);}
}

\title{A tensor version of the quantum Wielandt theorem}
\author{Mateusz Micha{\l}ek}
\address{Max Planck Institute for Mathematics in the Sciences\\ 
Leipzig, Germany\\
and\\
Polish Academy of Sciences, 
Institute of Mathematics\\ 
Warsaw, Poland\\
and\\
Department of Mathematics and Systems Analysis at Aalto University, Espoo, Finland}
\email{michalek@mis.mpg.de}
\thanks{MM was supported by Polish National Science Center project 2013/08/A/ST1/00804 affiliated at the University of Warsaw.}
\author{Tim Seynnaeve}
\address{Max Planck Institute for Mathematics in the Sciences\\ 
	Leipzig, Germany}
\email{tim.seynnaeve@mis.mpg.de}
\author{Frank Verstraete}
\address{Quantum Optics, Quantum Nanophysics and Quantum Information, Faculty of Physics\\ 
	University of Vienna, Vienna, Austria\\
and\\
Department of Physics and Astronomy\\
Ghent University, Ghent, Belgium}
\email{Frank.Verstraete@UGent.be}

\subjclass[2010]{15A69, 82B20}

\begin{document}


\begin{abstract}
	We prove boundedness results for the injectivity regions for PEPS. Our result is a higher-dimensional generalization of the quantum Wielandt inequality. 
\end{abstract}

\maketitle

\section{Introduction}

In \cite{QuantumWielandt}, Sanz, Perez-Garcia, Wolf and Cirac proved a quantum version of the Wielandt inequality. This theorem was motivated by the study of Matrix Product States and conjectures stated in \cite{PerezGarciaVerstraete}. Quantum Wielandt theorem proved upper bounds on the support of parent Hamiltonians for injective matrix product states, which was the final piece missing for proving that the manifold of MPS is in one to one correspondence with ground states of local Hamiltonians \cite{fannes1992finitely}.\\ 
In mathematical terms, the quantum Wielandt theorem can be stated as follows:\\ 
Let $\mathcal{A}=(A_1,\ldots,A_d)$ be a $d$-tuple of $D \times D$-matrices, and assume that there is an $N$ such that the linear span of $\{A_{i_1}\cdots A_{i_N} | 1 \leq i_j \leq d\}$ equals the space of $D \times D$-matrices. Then already for $N=C(D,d):=(D^2-d+1)D^2$, the linear span of $\{A_{i_1}\cdots A_{i_N} | 1 \leq i_j \leq d\}$ equals the space of $D \times D$-matrices.
The bound $C(D,d)$  was recently improved to $O(D^2\log D)$ in \cite{MichalekShitov} and is conjectured to be $O(D^2)$ \cite{PerezGarciaVerstraete}. \\ 
The 2-dimensional generalizations of MPS are called projected entangled pair states (PEPS), and play a central role in the classification of the different quantum phases of spin systems defined on two-dimensional grids. PEPS are much more complex than MPS: just as MPS can be understood in terms of completely positive maps on matrices, PEPS deal with completely positive maps on tensors, for which no analogues of eigenvalue and singular value decompositions exist. It has been a long standing open question in the field of quantum tensor networks whether an analogue of the quantum Wielandt theorem exists for PEPS, which is the missing piece in proving that every PEPS has a parent Hamiltonian with finite support. This paper proves the existence of such a theorem, albeit in a weaker form than for MPS as the upper bound in nonconstructive. In physics terms, it is proven that the notion of injectivity for PEPS is well defined, in the sense that there is only a finite amount of blocking needed for the map from the virtual to the physical indices to become injective.\\
The bounds for quantum Wielandt theorem in \cite{QuantumWielandt, MichalekShitov, rahaman2018new} were obtained using explicit methods from linear algebra.
Our main new insight is the application of nonconstructive Noetherian arguments from non-linear algebra.
For matrix product states, this gives an easy proof of Conjecture 1 in \cite{PerezGarciaVerstraete}.

\section{PEPS and injective regions}
By a \emph{grid} we mean a triple $G=(\V,E_I,E_O)$ where $\V,E_I,E_O$ are finite sets, respectively called vertices, inner edges and outgoing edges. We will write $E=E_I \cup E_O$. One can think of a grid as (a part of) a system of particles, where the vertices are the particles and the edges indicate interaction between the particles. Every inner edge $e$ may be identified with a two element subset $\{v_1,v_2\}$ of $\V$. Every outgoing edge distinguishes one element of $\V$, which we call its endpoint. For $v$ a vertex, $e(v)$ will denote the set of edges incident with that vertex. 
Let $G$ be a grid. To every edge $e$ of $G$ we associate a finite-dimensional $\mathbb{C}$-vector space $V_e$, equipped with an inner product.
 To every vertex $v$ of $G$, we associate two $\mathbb{C}$-vector spaces: the \emph{virtual space} $V_v:=\bigotimes_{e \in e(v)} V_e$, and the \emph{physical space} $W_v$, which represents the physical state of the particle.\\
  Let $(v \mapsto A_v)_{v \in G}$ be a function that assigns to every vertex $v$ of $G$ a tensor $A_v \in V_v$. We then obtain a new tensor in $\bigotimes_{e\in E_O}V_e$ by contracting along all inner edges of $G$. This tensor will be denoted by $\mathcal{C}[(v \mapsto A_v)_{v \in G}]$, or simply by $\mathcal{C}[v \mapsto A_v]$.
 \begin{example}
 	Multiplication of matrices.
 	Let $G$ be the following grid: \vskip-1cm
 	\[ \picH \]
 	with vertices $1$ and $2$.
 	Then $\mathcal{C}[(v \mapsto A_v)_{v \in G}]$ is simply the matrix product $A_1A_2$.

 \end{example}
\noindent If we assign to every vertex $v$ a linear map $\phi_v: V_v \to W_v$, we obtain a linear map
 \[
 \bigotimes_{e \in E_O}{V_e} \to \bigotimes_{e \in E_O}{V_e} \otimes  \bigotimes_{e \in E_I}{(V_e \otimes V_e)} \cong \bigotimes_{v\in \V}{V_v} \to \bigotimes_{v\in \V}{W_v}
 \]
 where the first map is induced by the inner product on $V_e$. If our grid has no outgoing edges, this corresponds to a physical state of our system. Any state that arises in this way is called a \emph{projected entangled pair state (PEPS)}.\\
 We now give a description in coordinates: write $d_v= \dim W_v$, $D_e = \dim V_e$, and suppose $\phi_v$ is given by
 \[
 \sum_{i=1}^{d_v}\sum_{\boldsymbol{j}}{A^{(v)}_{i,\boldsymbol{j}} \ket{i}\bra{\boldsymbol{j}}}
 \]
 where the second sum is over all tuples $\boldsymbol{j}=(j_e)_{e \in e(v)}$, $1 \leq j_e \leq D_e$. Then our map is 
 \[
\sum_{i_1,i_2,\ldots =1}^{d}{\mathcal{C}[v \mapsto A^{(v)}_{i_v}]}^*\ket{i_1,i_2,\ldots} \text{.}
 \]
 \begin{definition}
 If the above map is injective, we say that $(G,\{\phi_v\}_{v \in \V})$ is an \emph{injective region}. Equivalently, $(G,\{\phi_v\}_{v \in \V})$ is an injective region if and only if the tensors ${\mathcal{C}[v \mapsto A^{(v)}_{i_v}]}$ span the whole space $\bigotimes_{e \in E_O}{V_e}$.
 \end{definition}

 
\section{Main theorem}
We fix natural numbers $n$ (grid dimension), $D$ (bond dimension), $d$ (physical dimension). \\
For $N_1,\ldots,N_n \in \mathbb{N}$, we let $G=G(N_1,\ldots,N_n)$ be the $n$-dimensional square grid of size $N_1 \times \ldots \times N_n$. In particular, every vertex has degree $2n$. 
The grid $G(3,5)$ is presented below:
$$\picG .$$
We will denote the outgoing edges of $G$ by $(\boldsymbol{j},\pm \be_i)$, where $\boldsymbol{j}$ is a vertex on the boundary of the grid, and $\pm \be_i$ indicates the direction of the outgoing edge.\\
To every edge $e$ we associate the vector space $V_e = V = \mathbb{C}^D$. We stress that $D$ is now the same for every edge. 
Now we can identify all virtual spaces $V_v=\bigotimes_{e\in e(v)} V_e=(\mathbb{C}^D)^{\ot \deg(v)}=(\mathbb{C}^D)^{\ot 2n}$ in the obvious way: the tensor factor of $(\mathbb{C}^D)^{\ot \deg(v)}$ associated to an edge out of $v$ will be identified with the tensor factor of $(\mathbb{C}^D)^{\ot \deg(w)}$ associated to the edge out of $w$ pointing in the same direction.
We also identify all physical spaces $W_v$ with a fixed $d$-dimensional vector space.\\
Having done these identifications, we can now associate the same linear map  $\phi_{\mathcal{A}} = \sum_{i=1}^{d}\sum_{\boldsymbol{j}}{A_{i,\boldsymbol{j}} \ket{i}\bra{\boldsymbol{j}}}$, where $\mathcal{A} = (A_1,\ldots,A_d)$ be a collection of $d$ tensors $A_i \in (\mathbb{C}^D)^{\otimes2n}$. From now on, we will assume we have a fixed $\mathcal{A}$, and let the size of the grid $G=G(N_1,\ldots,N_n)$ vary.
\begin{definition}
	We say that \emph{$G$ is an injective region for $\mathcal{A}$} if $(G,\{\phi_{\mathcal{A}}\}_{v \in \V})$ in an injective region.
	Explicitly, $G$ is an injective region for $\mathcal{A}$, if the tensors $\mathcal{C}[(v \mapsto A_{i_v})_{v \in G}]$ linearly span whole space $(\mathbb{C}^D)^{\ot E_O(G)}$, when we consider all possible ways of placing a tensor from $\mathcal{A}$ at every vertex of $G$.\\
If $\mathcal{A}$ has an injective region, we say that $\mathcal{A}$ is \emph{injective}.
\end{definition}
\begin{remark}\label{rem:subspace}
We note that being an injective region for $G$ and being injective are properties of the linear span of $\mathcal{A}$, not a particular choice of tensors $A_i$.
\end{remark}
In the following lemma we prove that being an injective region is stable under extension of the grid.
\begin{lemma} \label{lemma:biggergrid}
	Let $G_1 \subseteq G_2$ be $n$-dimensional square grids. If $G_1$ is an injective region for $\mathcal{A}$, then so is $G_2$.
\end{lemma}
\begin{proof}
	By induction, we may assume that $G_1=G(N_1-1,N_2,\ldots,N_n)$ and $G_2=G(N_1,N_2,\ldots,N_n)$.  If $N_1=2$ the statement is true, because $G_2$ is the union of two injective regions, cf \cite[Lemma 1]{PEPS}. Thus we assume $N_1>2$. 
	The vertices of $G_2$ will be identified with vectors $\boldsymbol{j} = (j_1,\ldots,j_n) \in \mathbb{N}^n$, with $1 \leq j_i \leq N_i$. Such a vertex is in $G_1$ if additionally $N_1 \leq N_1-1$.
	We need to show that every tensor $T \in V^{\ot E_O(G_2)}$ can be written as a linear combination of tensors of the form $\mathcal{C}[(\boldsymbol{j} \mapsto A_{i_{\boldsymbol{j}}})_{\boldsymbol{j}\in G_2}]$. In fact it is enough to show this for rank one tensors $T$, since every tensor is a linear combination of rank one tensors.\\
	We can identify $E_O(G_1)$ with a subset of $E_O(G_2)$ as follows: to an outgoing edge $(\boldsymbol{j},\pm \be_i)$ of $E_O(G_1)$, we associate $(\boldsymbol{j},\pm \be_i)$ if $\pm \be_i \neq +\be_1$, and $(\boldsymbol{j}+\be_1,+\be_1)$ if $\pm \be_i= +\be_1$. Assuming $T$ has rank one, we have $T=T_1 \ot T_2 \in V^{\ot E_O(G_1)} \ot V^{\ot r}$, where $r$ equals the cardinality of $E_O(G(N_2,\dots,N_n))$.\\
	By assumption we can write $T_1$ as a linear combination of tensors of the form $\mathcal{C}[(\boldsymbol{j} \mapsto A_{i_{\boldsymbol{j}}})_{\boldsymbol{j}\in G_1}]$. 
	Let $G_1'$ be the grid obtained from $G_1$ by contracting all inner edges among vertices $\boldsymbol{j}$ for which $j_1 >0$. In particular, all vertices with $j_1>0$ get identified to a vertex $v_2$ (resp.~$v_1$). 
	Then $T_1$ is in particular a linear combination of tensors of the form $\mathcal{C}[(v \mapsto B_v)_{v \in G'_1}]$, where $B_v=A_{i_v}$ if $v$ is one of the vertices that did not get contracted and $B_{v_1}$ are some special tensors. \\
Consider the tensors $B_{v_1}\otimes T_2\in V^{\otimes |E_O(G_1)|}$. By assumption each one is a linear combination of tensors of the form $\mathcal{C}[(\boldsymbol{j} \mapsto A_{k_{\boldsymbol{j}}})_{\boldsymbol{j}\in G_1}]$, where now we identified $G_1$ with the subgrid of $G_2$ consisting of all vertices $\boldsymbol{j}$ with $j_1 > 0$.
	
Thus, we see that $T$ is a combination $\mathcal{C}[(\boldsymbol{j} \mapsto A_{s_{\boldsymbol{j}}})_{\boldsymbol{j}\in G_2}]$
where $s$ may be identified with $i$ above for $\boldsymbol{j}$ such that $j_1=0$ and with $k$ when $j_1>0$. The pictorial representation of the proof can be found below, where to each small box some element of $\mathcal{A}$ is associated.
\end{proof}

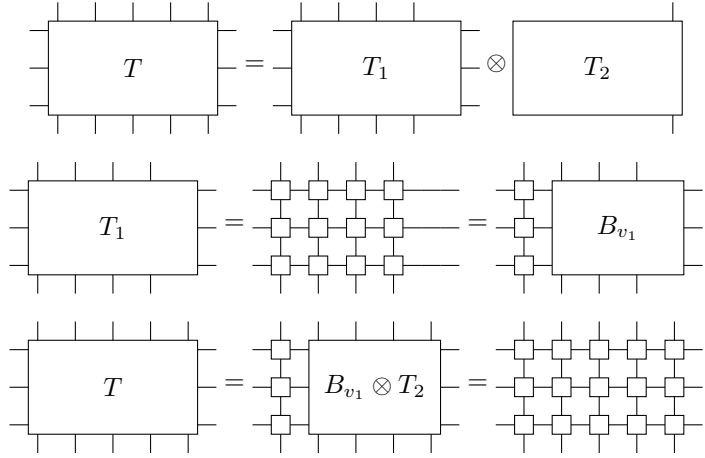
\[
\picA = \picB \ot \picC
\]

\[
\picB=\picD=\picE
\]

\[
\picA = \picF = \picG
\]

Our main theorem says that if $\mathcal{A}$ is injective, then there exists an injective region of bounded size (where the bound only depends on our parameters $D,d,n$). More precisely:
\begin{theorem} \label{mainthm}
	Let $G_1\subset G_2 \subset \cdots \subset G_k \subset \cdots$ be a chain of $n$-dimensional grids. Then there exists a constant $C$ (depending on $D$ and $d$) such that the following holds:
	If $\mathcal{A} \in ((\mathbb{C}^D)^{\otimes2n})^d$ is chosen so that for some $k$, $G_k$ is an injective region for $\mathcal{A}$, then already $G_C$ is an injective region for $\mathcal{A}$.
\end{theorem}
\begin{proof}
	For any grid $G$ and $\mathcal{A} \in ((\mathbb{C}^D)^{\otimes2n})^d$ , we write $S_G(\mathcal{A}) := \{\mathcal{C}[(v \mapsto A_{i_v})_{v \in G}]\} \subseteq (\mathbb{C}^D)^{\ot E_O(G)}$, and $V_G := \{\mathcal{A} \in ((\mathbb{C}^D)^{\otimes2n})^d | \Span(S_G(\mathcal{A}))  \subsetneq (\mathbb{C}^D)^{\ot E_O(G)} \} $. Thus, $G$ is an injective region for $\mathcal{A}$ if and only if $\Span(S_G(\mathcal{A})) = (\mathbb{C}^D)^{\ot E_O(G)}$ if and only if $\mathcal{A} \notin V_G$.\\
	By Lemma \ref{lemma:biggergrid}, it holds that $V_{G_1} \supseteq V_{G_2} \supseteq  \cdots \supseteq V_{G_k}  \supseteq \cdots$
	We need to to show that this chain eventually stabilizes. We will show that every $V_{G_k}$ is a Zariski closed subset of  $((\mathbb{C}^D)^{\otimes2n})^d$, i.e.\ that is is the zero locus of a system of polynomials. This will finish the proof by Hilbert Basis Theorem. \\ 
	Fix a grid $G=G_k$. For every $\mathcal{A} \in ((\mathbb{C}^D)^{\otimes2n})^d$, we can build a $D^{E_O(G)} \times d^{|\mathcal{V}(G)|}$ matrix $M_\mathcal{A}$ whose entries are the coefficients of the elements of $S_G(\mathcal{A})$. The condition $\Span(S_G(\mathcal{A}))  \subsetneq (\mathbb{C}^D)^{\ot E_O(G)}$ is equivalent to $M_{\mathcal{A}}$ having rank smaller than $D^{E_O(G)}$. The entries of $M_{\mathcal{A}}$ are polynomials in the entries of $\mathcal{A}$. Hence, the condition $\mathcal{A} \in V_G$ can be expressed as the vanishing of certain polynomials ($D^{E_O(G)}$-minors of $M_\mathcal{A}$) in the entries of $\mathcal{A}$. Hence, $V_G$ is a Zariski closed subset of  $((\mathbb{C}^D)^{\otimes2n})^d$.
\end{proof}

Theorem \ref{mainthm} can be reformulated as follows:
\begin{theorem} \label{mainthm2}
	There exists a finite collection of grids $G_1,\ldots,G_M$ (depending on $n,D,d$) such that the following holds: \\
	If $\mathcal{A} \in ((\mathbb{C}^D)^{\otimes2n})^d$ is injective, then one of the $G_i$ is an injective region for $\mathcal{A}$.
\end{theorem}

The equivalence of Theorem \ref{mainthm} and Theorem \ref{mainthm2} follows from the following general lemma:
\begin{lemma} \label{lemma:eq}
	Let $\mathcal{P}$ be a partially ordered set. We consider $\mathbb{N}^n$ with the coordinatewise partial order. Let $f: \mathbb{N}^n \to \mathcal{P}$ be a map such a that 
	\begin{enumerate}
		\item $\boldsymbol{a}_1 \leq \boldsymbol{a}_2 \implies f(\boldsymbol{a}_1) \geq f(\boldsymbol{a}_2)$.
		\item For every chain $\boldsymbol{a}_1 < \boldsymbol{a}_2 < \ldots$ in $\mathbb{N}^n$, the chain $f(\boldsymbol{a}_1) \geq f(\boldsymbol{a}_2) \geq \ldots$ stabilizes after finitely many steps. 
	\end{enumerate}
	Then there is a finite subset $B$ of $\mathbb{N}^n$ such that for any $\boldsymbol{a} \in \mathbb{N}^n$, there is a $\boldsymbol{b} \in B$ with $\boldsymbol{a} \geq \boldsymbol{b}$ and $f(\boldsymbol{a})=f(\boldsymbol{b})$.
\end{lemma}
\begin{proof}
	We first claim that there is a $\boldsymbol{b}_0 \in \mathbb{N}^n$ such that $f(\boldsymbol{a})=f(\boldsymbol{b_0})$ for every $\boldsymbol{a} \geq \boldsymbol{b}_0$. Indeed, if there was no such $\boldsymbol{b}_0$ we could build an infinite chain $\boldsymbol{a}_1 < \boldsymbol{a}_2 < \ldots$ in $\mathbb{N}^n$ with $f(\boldsymbol{a}_1) > f(\boldsymbol{a}_2) > \ldots$\\
	Now we can proceed by induction on $n$: the subset $\{\boldsymbol{a} \in \mathbb{N}^n | \boldsymbol{a} \ngeq \boldsymbol{b} \}$ can be written as a finite union of hyperplanes, each of which can be identified with $\mathbb{N}^{n-1}$. By the induction hypothesis, in each such hyperplane $H \subset \mathbb{N}^n$ there is a finite subset $B_H \subset H$ such that for any  $\boldsymbol{a} \in H$, there is a $\boldsymbol{b} \in B_H$ with $\boldsymbol{a} \geq \boldsymbol{b}$ and $f(\boldsymbol{a})=f(\boldsymbol{b})$.\\
	We define $B$ as $b_0$ together with the union of all $B_H$.
\end{proof}
To deduce Theorem \ref{mainthm2} from Theorem \ref{mainthm}, we apply Lemma \ref{lemma:eq} by identifying $\mathbb{N}^n$ with the collection of $n$-dimensional grids and taking $\mathcal{P}$ to be the poset of subsets of $((\mathbb{C}^D)^{\otimes2n})^d$ ordered by inclusion.
We consider $f: G \mapsto V_G$, where $V_G$ was defined in the proof of Theorem \ref{mainthm}.

We note that the constants in Theorem \ref{mainthm} and \ref{mainthm2} can be chosen independent of $d$.
\begin{corollary}
For any $n$ and $D$	there exists a finite collection of grids $G_1,\ldots,G_M$ such that the following holds: \\
	For any $d$, if $\mathcal{A} \in ((\mathbb{C}^D)^{\otimes2n})^d$ is injective, then one of the $G_i$ is an injective region for $\mathcal{A}$.
\end{corollary}
\begin{proof}
By Remark \ref{rem:subspace} it is enough to consider the subspaces $\langle \mathcal{A}\rangle\subset((\mathbb{C}^D)^{\otimes2n})$. In particular the dimension of the subspaces is bounded by $D^{2n}$ and for each fixed dimension we obtain a finite number of grids by Theorem \ref{mainthm2}.
\end{proof}
Further we have the following computational implication.
\begin{corollary}
There exists an algorithm to decide if $\mathcal{A}$ is injective.
\end{corollary}
For the case of dimension $n=1$ our result proves the Conjecture \cite[Conjecture 1]{PerezGarciaVerstraete}, whose effective version was proved in \cite{QuantumWielandt}.
\bibliography{PEPSbib}{}
\bibliographystyle{plain}
\end{document}